\documentclass[11pt]{article}
\usepackage{amsfonts,amsmath,amsthm,amssymb}

\usepackage{comment, enumerate}
\usepackage[margin=1in]{geometry}
\usepackage{soul,color}
\usepackage{mathrsfs}
\usepackage{mathdots}
\usepackage{algorithm}
\usepackage{algpseudocode}
\usepackage{caption}
\usepackage{subcaption}
\usepackage[pagebackref]{hyperref}
\usepackage{multicol}
\usepackage[small,compact]{titlesec}
\usepackage{times}
\usepackage{paralist}
\usepackage{cleveref}

\usepackage{bbm}

\usepackage{url}

\usepackage{tabu}

\usepackage{xfrac}

\usepackage[normalem]{ulem}

\usepackage{algorithm}
\usepackage{algorithmicx}
\usepackage{algpseudocode}

\newcommand{\F}{\mathbb{F}}

\newcommand{\R}{\mathbb{R}}

\newcommand{\eps}{\varepsilon}

\theoremstyle{plain}\newtheorem{theorem}{Theorem}[section]
\newtheorem{lemma}[theorem]{Lemma}

\theoremstyle{definition}

\usepackage[utf8]{inputenc}

\title{Faster Walsh-Hadamard Transform and Matrix Multiplication over Finite Fields using Lookup Tables} 

\author{Josh Alman\footnote{Columbia University. \texttt{josh@cs.columbia.edu}.}}

\begin{document}

\maketitle

\begin{abstract}

We use lookup tables to design faster algorithms for important algebraic problems over finite fields. 
These faster algorithms, which only use arithmetic operations and lookup table operations, may help to explain the difficulty of determining the complexities of these important problems. 
Our results over a constant-sized finite field are as follows.

The Walsh-Hadamard transform of a vector of length $N$ can be computed using $O(N \log N / \log \log N)$ bit operations. This generalizes to any transform defined as a Kronecker power of a fixed matrix. By comparison, the Fast Walsh-Hadamard transform (similar to the Fast Fourier transform) uses $O(N \log N)$ arithmetic operations, which is believed to be optimal up to constant factors.

Any algebraic algorithm for multiplying two $N \times N$ matrices using $O(N^\omega)$ operations can be converted into an algorithm using $O(N^\omega / (\log N)^{\omega/2 - 1})$ bit operations. For example, Strassen's algorithm can be converted into an algorithm using $O(N^{2.81} / (\log N)^{0.4})$ bit operations. It remains an open problem with practical implications to determine the smallest constant $c$ such that Strassen's algorithm can be implemented to use $c \cdot N^{2.81} + o(N^{2.81})$ arithmetic operations; using a lookup table allows one to save a super-constant factor in bit operations.

\end{abstract}

\thispagestyle{empty}
\newpage
\setcounter{page}{1}

\section{Introduction}

When $N$ is a power of $2$, the Walsh-Hadamard transform is defined recursively by $H_2 = \begin{bmatrix}
1 & 1\\
1 & -1\\
\end{bmatrix},$ and $$H_N = \begin{bmatrix}
H_{N/2} & H_{N/2}\\
H_{N/2} & -H_{N/2}\\
\end{bmatrix}.$$ A common task in many areas of computation is to \emph{compute the length-$N$ Walsh-Hadamard transform}, i.e., given as input a length-$N$ vector $v \in \F^N$, compute the vector $H_N v$. The most straightforward algorithm would compute this using $O(N^2)$ arithmetic operations, but the fast Walsh-Hadamard transform (FWHT) algorithm can compute this using only $O(N \log N)$ arithmetic operations. It is widely believed that $\Omega(N \log N)$ arithmetic operations are necessary, and a substantial amount of work has gone into proving this in restricted arithmetic models\footnote{For instance, this is known to hold for arithmetic circuits with `bounded coefficients'; see e.g.,~\cite[{Section 3.3}]{lokam2009complexity}.}, and studying conjectures which would imply this\footnote{For instance, the now-refuted conjecture that the Walsh-Hadamard transform is rigid~\cite{alman2017probabilistic}.}; see, for instance, the survey~\cite{lokam2009complexity}.

A natural question arises: why restrict ourselves to arithmetic models of computation? The Walsh-Hadamard transform is commonly used in practice, and speedups using non-arithmetic operations could be impactful. Nonetheless, there has not been much work on non-arithmetic algorithms for the Walsh-Hadamard transform.

A related problem is matrix multiplication: given as input two $n \times n$ matrices, compute their product. The asymptotically fastest known algorithm is algebraic, and uses $O(n^{2.373})$ arithmetic operations. That said, non-algebraic algorithmic techniques, especially lookup tables, have been used to design more practical, `combinatorial' algorithms for a variant on this problem called Boolean matrix multiplication\footnote{In Boolean matrix multiplication, given two matrices whose entries are from $\{0,1\}$, our goal is to multiply them over the AND-OR semiring, or equivalently, to determine which entries of their product over $\R$ are nonzero (but not necessarily what nonzero values they take on).} since the work of~\cite{arlazarov1970economical}. These techniques save logarithmic factors over the straightforward $O(n^3)$ time algorithm -- the best algorithm along these lines runs in $n^3 \cdot \frac{(\log\log n)^{O(1)}}{(\log n)^3}$ time~\cite{yu2018improved} in the standard word RAM model -- but they are considered more practical than the algebraic algorithms (which save polynomial factors in $n$). Lookup tables have also been used in practice to approximately multiply matrices~\cite{jeon2020biqgemm,blalock2021multiplying}.

In this paper, we show that lookup tables can be used to speed up the asymptotically fastest known algorithms for both the Walsh-Hadamard transform and (exact) matrix multiplication \emph{over finite fields}. We will show, for instance, that only $o(N \log N)$ bit operations suffice to compute the length-$N$ Walsh-Hadamard transform over finite fields when we augment the arithmetic model to allow for lookup tables. This may help to explain the difficulty of proving lower bounds for this problem, and help to guide future work on arithmetic circuit lower bounds (since any lower bounds technique would need to fail when lookup tables are allowed).

We focus here on constant-sized finite fields, though our algorithms generalize to larger finite fields as well. Our algorithms are simple modifications of the usual recursive algorithms for solving these problems, and we describe below how they compare favorably to algorithms which are used in practice.

\subsection{Model of computation}

As discussed, these problems are typically studied in the arithmetic circuit model, wherein an algorithm may only perform arithmetic operations ($+,-,\times,\div$) over the field $\F$ applied to inputs and fixed constants, and the arithmetic complexity of the algorithm is the number of such operations. The asymptotically fastest known algorithms for the problems we study here all fit within this model. However, we would also like the ability to use lookup tables, so we consider a more general model: the \emph{bit operations model}.

The bit complexity of a RAM algorithm is the number of operations on bits performed by the algorithm. This is the natural model with the most direct comparison to the arithmetic model, since an algorithm with arithmetic complexity $T$ over a constant-sized finite field naturally has bit complexity $O(T)$. This model is often used as a more realistic version of the arithmetic model (see e.g.,~\cite{pan1981bit,lingas1991bit,van2013bit,el2018bit}).

One can see (via a simple tree data structure, for instance) that a lookup table with $b$-bit keys and values can be implemented so that values can be looked up and changed using $O(b)$ bit operations.

We note before moving on that all the algorithms in this paper will only perform arithmetic operations and lookup table operations; one could also define a nonstandard model of computation which allows for just these two types of operations, and get the same results.

\subsection{Results: Walsh-Hadamard Transform}

Our first result is a faster algorithm for the Walsh-Hadamard transform.

\begin{theorem}
Let $\F_q$ be a finite field of size $q = O(1)$, let $n$ be a positive integer, and let $N = 2^n$. There is an algorithm for computing the length-$N$ Walsh-Hadamard transform over $\F_q$ that uses $O\left(\frac{N \log N}{\log \log N}\right)$ bit operations.
\end{theorem}

By comparison, the fast Walsh-Hadamard transform algorithm, which uses $\Theta(N \log N)$ arithmetic operations, would take $\Theta(N \log N)$ bit operations. $\Theta(N \log N)$ arithmetic operations is widely believed to be optimal over any field (whose characteristic is not 2; the problem is trivial in that case), but our algorithm improves on this using lookup tables.

Our algorithm uses the same recursive approach as the fast Walsh-Hadamard transform; our main idea is to use results from a lookup table to quickly jump forward many recursive layers at a time.

Although the Walsh-Hadamard transform is most often applied over the real or complex numbers, such as in signal processing and data compression, 
it has been applied over finite fields in areas including coding theory~\cite{rajan2001quasicyclic,xu2019three}, cryptographic protocols~\cite{helleseth2010new,mesnager2019two}, and learning algorithms~\cite{liu2020deep}.

Our algorithm also generalizes directly to any transform defined by Kronecker powers of a fixed matrix over a finite field. For instance, given as input the $2^n$ coefficients of a multilinear polynomial in $n$ variables over $\F_2$, we can compute its evaluation on all $2^n$ inputs from $\F_2^n$ in $O(2^n \cdot n / \log n)$ bit operations (improving over the usual recursive algorithm by a factor of $\log n$)\footnote{This problem corresponds to computing the linear transform defined by Kronecker powers of the matrix $\begin{bmatrix}
1 & 1\\
1 & 0\\
\end{bmatrix}$.}.

\subsection{Results: Matrix Multiplication}

Our result for matrix multiplication shows how to convert any algebraic algorithm into one which uses a lookup table to save a superconstant factor.

\begin{theorem} \label{thm:mainmm}
For any finite field $\F_q$ of size $q = O(1)$, suppose there is an algebraic algorithm for multiplying $n \times n$ matrices over $\F_q$ in $O(n^{\tau})$ arithmetic operations for some constant $\tau > 2$. Then, there is another algorithm for multiplying $n \times n$ matrices over $\F_q$ which uses $O(n^\tau  / (\log n)^{\tau/2 - 1})$ bit operations.
\end{theorem}
This speeds up the standard implementation of the algebraic algorithm by a factor of $\Theta((\log n)^{\tau/2 - 1})$. For instance, Strassen's algorithm~\cite{strassen} gives $\tau = 2.81$, resulting in an algorithm using $O(n^{2.81} / (\log n)^{0.4})$ bit operations by Theorem~\ref{thm:mainmm}, and the asymptotically fastest known algorithm~\cite{coppersmith,alman2021refined} gives $\tau = 2.373$, resulting in an algorithm using $O(n^{2.373} / (\log n)^{0.186})$ bit operations by Theorem~\ref{thm:mainmm}.

Notably, much work has gone into improving the leading constant in the running time of Strassen's algorithm. Strassen's original algorithm~\cite{strassen} has leading constant $7$ (meaning, it uses $7 n^{\log_2(7)} + o(n^{\log_2(7)})$ operations), and Winograd~\cite{winograd1971multiplication} improved this to $6$. This was believed to be optimal due to lower bounds by Probert~\cite{probert1976additive} and Bshouty~\cite{bshouty1995additive}. However, in a recent breakthrough, Karstadt and Schwartz~\cite{karstadt2020matrix} gave a new `change of basis' approach, and used it to improve the leading constant to $5$. They showed that $5$ is optimal for their new approach, and later work showed it is also optimal for the more general `sparse decomposition' approach~\cite{beniamini2019faster}. The fact that we achieve an asymptotic speedup in this paper (which one can view as achieving leading constant $\eps$ for any $\eps>0$ in bit complexity) may help to explain the difficulty of extending these lower bounds on the constant beyond restricted classes of algorithmic approaches.

Our approach for matrix multiplication is one that is commonly used in practice: use iterations of a recursive algebraic algorithm until the matrices one would like to multiply are sufficiently small, and then use a different algorithm optimized for small matrices. When this is used in practice, an optimized version of the straightforward (cubic-time) algorithm is used for small matrices, giving a constant factor improvement to the running time; see e.g.,~\cite{huang2016strassen}. 
We implement this approach by instead using lookup tables to multiply superconstant-sized matrices very quickly, and thus get an asymptotic improvement.

Matrix multiplication over finite fields has many applications. One prominent example is Boolean matrix multiplication, which has a simple randomized reduction to matrix multiplication over any field\footnote{Set each entry of one of the matrices to 0 independently with probability $1/2$, then multiply the two matrices and check which entries are nonzero.}. Hence our algorithm gives an asymptotic speedup for applications of Boolean matrix multiplication such as detecting triangles in graphs.

\section{Preliminaries}

\subsection{Notation}

Throughout this paper, we write $\log$ to denote the base $2$ logarithm. For a positive integer $n$, we use the notation $[n] := \{1,2,3,\ldots,n\}$. 

\subsection{Kronecker Products}

If $\F$ is a field, and $A \in \F^{N \times N}$ and $B \in \F^{M \times M}$ are matrices, their \emph{Kronecker product} $A \otimes B \in \F^{(NM) \times (NM)}$ is a matrix given by
$$A \otimes B = \begin{bmatrix}
B[1,1] \cdot A & B[1,2] \cdot A & \cdots & B[1,M] \cdot A\\
B[2,1] \cdot A & B[2,2] \cdot A & \cdots & B[2,M] \cdot A\\
\vdots & \vdots & \ddots & \vdots \\
B[M,1] \cdot A & B[M,2] \cdot A & \cdots & B[M,M] \cdot A\\
\end{bmatrix}.$$
For a positive integer $n$, we write $A^{\otimes n} \in \F^{N^n \times N^n}$ for the $n$th \emph{Kronecker power} of $A$, which is the Kronecker product of $n$ copies of $A$.

\subsection{Matrices of interest}
For positive integer $N$, we write $I_N$ to denote the $N \times N$ identity matrix. If $N$ is a power of $2$, we write $H_N$ to denote the $N \times N$ Walsh-Hadamard transform, given by $H_N = H_2^{\otimes \log N}$, where $$H_2 = \begin{bmatrix}
1 & 1\\
1 & -1\\
\end{bmatrix}.$$

\section{Walsh-Hadamard Transform}

Let $\F_q$ be the finite field of size $q$, and let $N$ be a power of $2$. In this section, we give an algorithm for computing the length-$N$ Walsh-Hadamard transform, $H_N$, over $\F_q$. The key idea behind our new algorithm is to pick a $K = O(\log_q N)$ such that $q^K \ll N$, and first create a lookup table of the length-$K$ Walsh-Hadamard transforms of all vectors $v \in \F_q^K$. We will then use this lookup table in conjunction with the following standard recursive approach for computing Kronecker powers (sometimes called Yates' algorithm), which we will apply with $M = H_K$:

\begin{lemma} \label{lem:yates}
Let $\F_q$ be any constant-sized finite field, let $m,d$ be positive integers, and let $M \in \F^{d \times d}$ be any matrix. Suppose we are given an algorithm which, on input $v \in \F^d$, outputs $Mv$ in time $T$. Then, there is an algorithm which, on input $z \in \F^{d^m}$, outputs $M^{\otimes m} w$ in time $O(T \cdot m \cdot d^{m-1})$.
\end{lemma}

\begin{proof}
By definition of the Kronecker product, we can write $M^{\otimes m}$ as a $d \times d$ block matrix (where each block is a $d^{m-1} \times d^{m-1}$ matrix) as
\begin{align*} &M^{\otimes m} \\ &= \begin{bmatrix}
M[1,1] \cdot M^{\otimes (m-1)} & M[1,2] \cdot M^{\otimes (m-1)} & \cdots & M[1,d] \cdot M^{\otimes (m-1)}\\
M[2,1] \cdot M^{\otimes (m-1)} & M[2,2] \cdot M^{\otimes (m-1)} & \cdots & M[2,d] \cdot M^{\otimes (m-1)}\\
\vdots & \vdots & \ddots & \vdots \\
M[d,1] \cdot M^{\otimes (m-1)} & M[d,2] \cdot M^{\otimes (m-1)} & \cdots & M[d,d] \cdot M^{\otimes (m-1)}\\
\end{bmatrix}
\\ &= \begin{bmatrix}
M[1,1] \cdot I_{d^{m-1}} & M[1,2] \cdot I_{d^{m-1}} & \cdots & M[1,d] \cdot I_{d^{m-1}}\\
M[2,1] \cdot I_{d^{m-1}} & M[2,2] \cdot I_{d^{m-1}} & \cdots & M[2,d] \cdot I_{d^{m-1}}\\
\vdots & \vdots & \ddots & \vdots \\
M[d,1] \cdot I_{d^{m-1}} & M[d,2] \cdot I_{d^{m-1}} & \cdots & M[d,d] \cdot I_{d^{m-1}}\\
\end{bmatrix} \hspace{-4pt} \times \hspace{-4pt} \begin{bmatrix}
M^{\otimes (m-1)} & & \\
 & M^{\otimes (m-1)} &  & \\
 &  & \ddots &  \\
 &  &  &  M^{\otimes (m-1)}\\
\end{bmatrix}
\end{align*}
Thus, we can multiply $M^{\otimes m}$ times $z \in \F^{d^m}$ with a two-step process (corresponding to multiplying the matrix on the right times $z$, then the matrix on the left times the result):
\begin{enumerate}
    \item Partition $z \in \F^{d^m}$ into $d$ vectors $z_1, \ldots, z_d \in \F^{d^{m-1}}$. Recursively compute $u_i = M^{\otimes (m-1)} z_i \in \F^{d^{m-1}}$ for each $i \in \{ 1, \ldots, d \}$. 
    \item For each $j \in \{ 1, \ldots, d^{m-1} \}$, let $x_j \in \F^d$ be the vector consisting of, for each $i \in \{1, \ldots, d\}$, entry $j$ of vector $u_i$. Use the given algorithm to compute $y_j = M  x_j$.
\end{enumerate}
Finally, we output the appropriate concatenation of the $y_j$ vectors (where the first $d^{m-1}$ entries are the first entries of all the $y_j$ vectors, the second $d^{m-1}$ entries are the second entries of all the $y_j$ vectors, and so on).

Our algorithm makes $d$ recursive calls in the first step, and calls the given algorithm $d^{m-1}$ times in the second step.
Hence, the total running time, $E(d^m)$, has the recurrence $$E(d^m) = d \cdot E(d^{m-1}) + d^{m-1} \cdot T.$$
This solves, as desired, to $E(d^m) = O(T \cdot m \cdot d^{m-1})$.
\end{proof}

We can now give our main algorithm for the Walsh-Hadamard transform:

\begin{theorem} \label{thm:bodywht}
Let $\F_q$ be a finite field of size $q = O(1)$, let $n$ be a positive integer, and let $N = 2^n$. There is an algorithm for computing the length-$N$ Walsh-Hadamard transform over $\F_q$ that uses $O\left(\frac{N \log N}{\log \log N}\right)$ bit operations.
\end{theorem}

\begin{proof}
Let $k = \left\lfloor \log \left( \frac{n}{2 \log q} \right) \right\rfloor$, and let $K = 2^k \leq \frac{n}{2 \log q}$. We begin by iterating over all vectors $v \in \F_q^K$, computing $H_K v$, and storing it in a lookup table. We can do this in a straightforward way: there are $q^K$ such vectors, and each can be computed using $O(K^2)$ additions and subtractions over $\F_q$, so the total time to create this lookup table is at most $$O(q^K \cdot K^2) \leq O(q^{n / (2 \log q)} \cdot (\log N)^2) \leq O(\sqrt{N} \cdot (\log N)^2)  \leq O(N^{0.6} ).$$ (This simple time bound could be improved, but we won't bother here since it won't substantially contribute to the final running time.)

Our goal now is, given $v \in \F_q^N$, to compute $H_N v$. Assume for now that $k$ divides $n$.
We will apply \Cref{lem:yates} with $M = H_K$ (and hence $d = K$) and $m = n/k$, which will multiply the matrix $(H_K)^{\otimes n/k} = H_N$ times the vector $v$, as desired. Each time that algorithm needs to multiply $H_K$ times a vector of length $K$, we do so by looking up the answer from the lookup table. Hence, $T$ in \Cref{lem:yates} will be the time to do one lookup from this table whose keys and values have length $O(\log(q^K)) = O(K)$, so $T = O(K)$.

The total number of bit operations of our algorithm is thus, as desired, $$O\left(N^{0.6}  + T \cdot \frac{n}{k} \cdot K^{n/k - 1}\right) = O\left(N^{0.6}  + T \cdot \frac{n}{k \cdot K} \cdot N\right) = O\left(N^{0.6}  + \frac{n}{k} \cdot N\right) = O\left(\frac{N \log N }{\log \log N}\right).$$

Finally, consider when $k$ does not divide $n$. Let $n'$ be the largest multiple of $k$ that is smaller than $n$, so $n-k < n' < n$. By the usual recursive approach (e.g., one recursive step of the algorithm presented in \Cref{lem:yates}), it suffices to first perform $2^{n - n'}$ instances of a length-$2^{n'}$ Walsh-Hadamard transform, and then perform $2^{n'}$ instances of a length-$2^{n-n'}$ Walsh-Hadamard transform.
We now count the number of bit operations for these two steps.

Using the same algorithm as above, since $k$ divides $n'$, a length-$2^{n'}$ Walsh-Hadamard transform can be performed using $O(\frac{n' \cdot 2^{n'} }{\log n'})$ bit operations. Hence, the total bit operations for the first step is $O(2^{n - n'} \cdot \frac{n' \cdot 2^{n'} }{\log n'}) \leq O(\frac{n \cdot 2^{n} }{\log n}) = O(N \log N / \log \log N)$.

Using the usual fast Walsh-Hadamard transform, a length-$2^{n-n'}$ Walsh-Hadamard transform can be performed using $O(2^{n-n'} \cdot (n - n') )$ bit operations. Hence, the total bit operations for the second step is $O(2^{n'} \cdot 2^{n-n'} \cdot (n - n'))\leq O(2^n \cdot k) = O(N \log \log N)$. We thus get the desired total running time.
\end{proof}

\section{Matrix Multiplication}

Fast algebraic algorithms for matrix multiplication over a field $\F$ critically rely on algebraic identities which take the following form, for positive integers $q,r$, formal variables $X_{i,j}, Y_{j,k}, Z_{i,k}$ for $i,j,k \in [t]$, and field coefficients $\alpha_{i,j,\ell}, \beta_{j,k,\ell}, \gamma_{i,k,\ell} \in \F$ for $i,j,k \in [t]$ and $\ell \in [r]$:

\begin{align}\label{eq:mm}\sum_{i=1}^t \sum_{j=1}^t \sum_{k=1}^t X_{i,j} Y_{j,k} Z_{i,k} = \sum_{\ell=1}^r \left( \sum_{i=1}^t \sum_{j=1}^t \alpha_{i,j,\ell} X_{i,j} \right)\left( \sum_{j=1}^t \sum_{k=1}^t \beta_{j,k,\ell} Y_{j,k} \right)\left(  \sum_{i=1}^t \sum_{k=1}^t \gamma_{i,k,\ell} Z_{i,k} \right).\end{align}

As we will see next, an identity (\ref{eq:mm}) can be used to design a matrix multiplication algorithm which runs using only $O(n^{\log_t(r)})$ field operations. 
For instance, Strassen's algorithm~\cite{strassen} gives an identity (\ref{eq:mm}) with $t=2$ and $r=7$, yielding exponent $\log_2(7) < 2.81$. The matrix multiplication exponent $\omega$ is defined as the infimum of all numbers such that,  for every $\eps>0$, there is a sufficiently large $t$ for which one gets an identity (\ref{eq:mm}) with $r \leq t^{\omega + \eps}$. 
Indeed, it is known~\cite{strassen1973vermeidung} that any algebraic algorithm for matrix multiplication can be converted into an identity (\ref{eq:mm}) yielding the same running time in this way, up to a constant factor, so $\omega$ captures the best exponent from any possible algebraic algorithm. The fastest known matrix multiplication algorithm~\cite{coppersmith,alman2021refined} shows that $\omega < 2.37286$.

The standard recursive algorithm for multiplying matrices using identity (\ref{eq:mm}) works as described in \Cref{alg:mm} below. It recurses until it gets to a base case of multiplying $n \times n$ matrices when $n \leq S$ for some parameter $S$. In the usual algorithm, one picks $S$ to be a constant, so that such matrices can be multiplied in a constant number of operations; in our improvement, we will pick a larger $S$ and multiply such matrices using a lookup table. 

\begin{algorithm}[H]\caption{Recursive matrix multiplication algorithm, using identity (\ref{eq:mm}), with base case size $S$}\label{alg:mm}
\begin{algorithmic}[1]
\Procedure{\textsc{MM}}{$X, Y \in \F^{n \times n}$}
\If{$n \leq S$}
\State Use base case procedure to multiply $X \times Y$ and output the result.
\Else
\State Partition $X$ into a $t \times t$ block matrix, with blocks $X_{i,j} \in \F^{(n/t) \times (n/t)}$ for $i, j \in [t]$.
\State Similarly partition $Y$ into a $t \times t$ block matrix, with blocks $Y_{j,k} \in \F^{(n/t) \times (n/t)}$ for $j, k \in [t]$.
\State For $k, i \in [t]$, set the matrices $Z_{k,i} \in \F^{(n/t) \times (n/t)}$ to initially be all 0.
\For{$\ell \in [r]$}
\State Compute $A_\ell := \sum_{i, j \in [t]} \alpha_{i,j,\ell} X_{i,j}$.
\State Compute $B_\ell := \sum_{j, k \in [t]} \beta_{j,k,\ell} Y_{j,k}$.
\State Compute $C_\ell = \textsc{MM}(A_\ell, B_\ell)$ \Comment{Recursively multiply $(n/t) \times (n/t)$ matrices $A_\ell \times B_\ell$.}
\For{$i, k \in [t]$}
\State Add $\gamma_{i,k,\ell} C_\ell$ to $Z_{i,k}$.
\EndFor
\EndFor
\State Output $Z \in \F^{n \times n}$ given by the blocks $Z_{i,k}$.
\EndIf
\EndProcedure
\end{algorithmic}
\end{algorithm}

\begin{lemma}
\Cref{alg:mm} correctly outputs the product $X \times Y$.
\end{lemma}

\begin{proof}
For a fixed $i, k \in [t]$, we can see that lines 8-15 of the algorithm will set $Z_{i,k}$ to be the following matrix: $$\sum_{\ell=1}^r \gamma_{i,k,\ell} \cdot C_\ell = \sum_{\ell=1}^r \gamma_{i,k,\ell} \cdot \left( \sum_{i=1}^t \sum_{j=1}^t \alpha_{i,j,\ell} X_{i,j} \right)\left( \sum_{j=1}^t \sum_{k=1}^t \beta_{j,k,\ell} Y_{j,k} \right).$$
Notice in particular that this is the coefficient of the formal variable $Z_{i,k}$ in the right-hand side of the identity~(\ref{eq:mm}). Hence, it is also the coefficient of that variable in the left-hand side, namely, $$\sum_{j=1}^t X_{i,j} Y_{j,k}.$$ This is the desired output $(i,k)$ block when multiplying matrices $X,Y$.
\end{proof}

Normally one would analyze the recursion for the number of operations performed by this algorithm when $S$ is a constant, and conclude a total operation count of $O(n^{\log_t(r)})$. Here we will improve on this over finite fields using a lookup table:

\begin{theorem} \label{thm:mainmm}
Fix an identity~(\ref{eq:mm}) and let $\tau = \log_t(r)$.  
For any finite field $\F_q$ of size $q = O(1)$, and any positive integer $n$, there is an algorithm for multiplying $n \times n$ matrices over $\F_q$ which uses $O(n^\tau  / (\log n)^{\tau/2 - 1})$ bit operations.
\end{theorem}

\begin{proof}
Let $s = \left\lfloor \sqrt{ \frac14 \log n / \log q} \right\rfloor$. We begin by iterating over all positive integers $s' \leq s$ and all pairs of $s' \times s'$ matrices $A, B \in \F_q^{s' \times s'}$, computing their product $A \times B$, and storing it in another lookup table. Similar to before, the running time for this step won't substantially contribute to the final running time, so we give a straightforward upper bound on the time it takes. The number of pairs of matrices we need to multiply is at most $s \cdot (q^{s^2})^2$, and the time to multiply each pair is at most $O(s^3 )$ by the straightforward algorithm. The total time to create the lookup table is hence at most $$ O(s^4 \cdot (q^{s^2})^2) \leq O(q^{\frac12 \log n / \log q} \cdot (\log n)^2) \leq O(n^{0.5} \cdot (\log n)^2).$$
Note that this table's keys and values are strings of length $O(\log (q^{s^2})^2) \leq O(\log n)$, so lookup table operations can be performed using $O(\log n)$ bit operations.

We now use \Cref{alg:mm}, with base case size $S = s$. The base case procedure is that, whenever we need to multiply two $s' \times s'$ matrices for $s' \leq S$, we find the result in the lookup table.

Let $E(m)$ denote the running time of this algorithm to multiply two $m \times m$ matrices. We get $E(m) = O(\log n)$ if $m \leq s$, and if $m > s$, the recurrence \begin{align}\label{eq:recurr}E(m) \leq r \cdot E(m/t) + O(m^2),\end{align}
recalling that $r,t$ are constants given in the identity~(\ref{eq:mm}). The $O(m^2)$ term in the right-hand side of Equation~(\ref{eq:recurr}) counts the bit operations to do a constant number of additions and scalar multiplications of $m/t \times m/t$ matrices. Solving this recurrence yields $$E(n) = O\left( \left( \frac{n}{s} \right)^{\log_t(r)} \cdot \log n \right) = O\left( \frac{n^\tau }{ (\log n)^{\tau/2 - 1}} \right),$$ as desired.
\end{proof}

\section{Conclusion}
We showed that for two important open problems -- determining the complexity of the Walsh-Hadamard transform, and determining the leading constant of Strassen's algorithm for matrix multiplication -- asymptotic improvements over the conjectured optimal arithmetic bounds are possible if one is allowed to use bit operations rather than just arithmetic operations. 

Our algorithms only made use of arithmetic operations and lookup table operations, so they could extend to other models of computation as well. One natural question is whether they extend to the standard word RAM model with word size $w = O(\log n)$ for input size $n$. Indeed, operations for the lookup tables we use, (with keys and values of $O(\log n)$ bits) require only $O(1)$ word operations in this model.

It is not hard to see that our algorithm for matrix multiplication can be implemented to take advantage of this model, improving Theorem~\ref{thm:mainmm} to running time $O(n^\tau / (\log n)^{\tau / 2})$ (improving by a factor of $\log n$). 

On the other hand, our algorithm for the Walsh-Hadamard transform seemingly \emph{cannot} be implemented in the word RAM model to get asymptotic savings. The main culprit is step 2 in the proof of Lemma~\ref{lem:yates}: if we want to efficiently use the lookup table to compute $y_j = M x_j$, then we first have to permute the bits of the $u_i$ vectors so that each $x_j$ fits in a single word. In other words, we are given the $u_i$ vectors in order, and would like to permute their entries to get the $x_j$ vectors in order instead.

Let $s = \Theta(d^m / w)$ be the number of words that our input fits into.
In general it is known that performing a fixed permutation of the bits of a string contained in $s$ words, for $s \geq w$, requires $\Theta(s \log s)$ word operations~\cite{brodnik1997trans}. However, our particular permutation can be broken up into $s/w$ different permutations on $w$ words (e.g., the first words in the descriptions of $x_1, \ldots, x_w$ are a permutation of the first words in the descriptions of $u_1, \ldots, u_w$). It can thus be performed in only $O(s \log w)$ word operations.

Since this must be done at all $m$ levels of recursion, it incurs an additional $\Theta(m \cdot d^m \cdot \frac{\log w}{w})$ word operations. With the parameter setting we ultimately use in Theorem~\ref{thm:bodywht}, this is $\Theta(N)$ word operations to perform all these permutations. By comparison, it is not too difficult to implement the usual fast Walsh-Hadamard transform to use a total of $O(N)$ word operations as well. Hence, the time to perform permutations (which doesn't come into play in the arithmetic or bit operation models) swamps any other computational savings in this model, and another approach is needed for a speedup.

\section*{Acknowledgements}
I would like to thank Dylan McKay and Ryan Williams for invaluable discussions throughout this project, and anonymous reviewers for helpful comments. 
This research was supported in part by a grant from the Simons Foundation (Grant Number 825870 JA).

\bibliographystyle{alpha}
\bibliography{papers}

\end{document}